\documentclass[conference,letterpaper]{IEEEtran}

\usepackage[utf8]{inputenc} 
\usepackage[T1]{fontenc}
\usepackage{url}
\usepackage{ifthen}
\usepackage{cite}
\usepackage[cmex10]{amsmath} 

\usepackage{cite}
\usepackage{amsmath,amssymb,amsfonts}
\usepackage{mathtools}
\usepackage{amsthm}
\usepackage{algorithmic}
\usepackage{graphicx}
\usepackage{textcomp}
\usepackage{tikz}
\usepackage{caption}
\usepackage{cuted}
\usepackage{romannum}
\usepackage[utf8]{inputenc}
\usepackage{pgfplots} 
\usepackage{pgfgantt}
\usepackage{pdflscape}
\usepackage{amssymb}
\usepackage{comment}
\usepackage{pst-plot}
\usetikzlibrary{spy}
\usetikzlibrary{positioning,calc}
\usetikzlibrary{decorations.pathmorphing,calc,shapes,shapes.geometric,patterns}
\usetikzlibrary{shapes.multipart}
\usepackage{xfrac}
\usepackage{colortbl}
\usepackage{cancel}
\usetikzlibrary{arrows,positioning,calc,intersections}
\usetikzlibrary{datavisualization.formats.functions}
\def\BibTeX{{\rm B\kern-.05em{\sc i\kern-.025em b}\kern-.08em
    T\kern-.1667em\lower.7ex\hbox{E}\kern-.125emX}}
    
\usepackage{pgfplots}
\usepgfplotslibrary{fillbetween}
\usetikzlibrary{arrows, decorations.markings}
\newtheorem{theorem}{Theorem}
\newtheorem{lemma}{Lemma}

\newtheorem{definition}{Definition}

\newcommand{\seta}{\ensuremath{\mathcal{A}}}

\newcommand{\setd}{\ensuremath{\mathcal{D}}}

\newcommand{\bs}[1]{\boldsymbol{#1}}

\newcommand{\RM}[1]{\MakeLowercase{\romannumeral #1{}}}

\definecolor{calpolypomonagreen}{rgb}{0.12, 0.3, 0.17}
\newcounter{remarkcount}
\newenvironment{remark}{\refstepcounter{remarkcount}\begin{trivlist}\item \textbf{Remark \theremarkcount.}}{\end{trivlist}}
\newcommand{\circlearrow}{}
\DeclareRobustCommand{\circlearrow}{%
  \mathrel{\vphantom{\rightarrow}\mathpalette\circle@arrow\relax}%
}
\newcommand{\circle@arrow}[2]{%
  \m@th
  \ooalign{%
    \hidewidth$#1\circ\mkern1mu$\hidewidth\cr
    $#1-$\cr}%
}
\makeatother
\let\emptyset\varnothing
\usepackage{amsmath, amssymb, amsfonts, amsthm}
\usepackage{bm}
\usepackage{xcolor}
\usepackage{framed}

\usepackage{pgf}
\usepackage{tikz}
\usetikzlibrary{arrows,automata}
\usetikzlibrary{positioning}
\usetikzlibrary{decorations.pathmorphing}
\usetikzlibrary{shapes.geometric}
\usetikzlibrary{fit}
\usetikzlibrary{backgrounds}

\usepackage[caption=false,font=footnotesize]{subfig}

\newcommand{\mbb}{\mathbb}

\theoremstyle{definition}

\theoremstyle{remark}

\usepackage{hyperref}
\begin{document}
\title{Common Randomness Generation over Slow Fading Channels} 

\author{
\IEEEauthorblockN{Rami Ezzine,\IEEEauthorrefmark{1}  Moritz Wiese,\IEEEauthorrefmark{1}\IEEEauthorrefmark{3} Christian Deppe \IEEEauthorrefmark{2} and Holger Boche \IEEEauthorrefmark{1}\IEEEauthorrefmark{3}\IEEEauthorrefmark{4}}
\IEEEauthorblockA{\IEEEauthorrefmark{1}Technical University of Munich, Chair of Theoretical Information Technology, Munich, Germany\\
\IEEEauthorrefmark{2}Technical University of Munich, Institute for Communications Engineering,  Munich, Germany\\
\IEEEauthorrefmark{3}CASA -- Cyber Security in the Age of Large-Scale Adversaries–
Exzellenzcluster, Ruhr-Universit\"at Bochum, Germany\\
\IEEEauthorrefmark{4}Munich Center for Quantum Science and Technology (MCQST), Schellingstr. 4, 80799 Munich, Germany\\
Email: \{rami.ezzine, wiese, christian.deppe, boche\}@tum.de}
}

\maketitle

\begin{abstract}
  This paper analyzes the problem of common randomness (CR) generation from correlated discrete sources aided by unidirectional communication over Single-Input Single-Output (SISO) slow fading channels with additive white Gaussian noise (AWGN) and arbitrary state distribution. Slow fading channels are practically relevant for wireless communications.
  We completely solve the SISO slow fading case by establishing its corresponding outage CR capacity using our characterization of its channel outage capacity. 
  The generated CR could be exploited to improve the performance gain in the identification scheme. The latter is known to be more efficient than the classical transmission scheme in many new applications, which demand ultra-reliable low latency communication.

\end{abstract}

\begin{IEEEkeywords}
Common randomness, slow fading, outage capacity.
\end{IEEEkeywords}
\section{Introduction}
Common randomness (CR) of two terminals refers to a random variable observable to both, with low error probability. In many models, one terminal corresponds to the sender station and the other corresponds to the receiver station. The availability of this CR  allows to implement correlated
random protocols leading to developing potentially faster and more efficient
algorithms \cite{corrsurvey}.
CR generation plays a major role in  sequential secret key generation   \cite{secretkey}. In the context of secret key generation, CR is usually denoted by Information reconciliation. CR is also highly relevant in the identification scheme, an approach in communications
developed by Ahlswede and Dueck \cite{Idchannels} in 1989. In the identification framework,
the encoder sends an identification message also called identity
over the channel and the decoder is not interested in what
the received message is. He wants to know if a specific message of special interest to him has
been sent or not. The identification scheme is better suited than the classical transmission scheme for
many new applications with high requirements on reliability and latency. These applications include several machine-to-machine
and human-to-machine systems \cite{application}, the tactile internet
\cite{Tactilesinternet}, digital watermarking \cite{MOULINwatermarking,AhlswedeWatermarking,SteinbergWatermarking}, industry 4.0 \cite{industry4.0}, molecular communications \cite{macromolecular1}\cite{macromolecular2}, etc. \color{black}Furthermore,\cite{PatentBA} describes an interesting application where identification codes \cite{implementation}  can be used in autonomous driving. This is a typical use case for ultra-reliable low latency communication\color{black}.
Interestingly, it has been established that the resource CR can increase the identification capacity of channels\cite{part2,Generaltheory,CRincrease}.
Thus, by taking advantage of the resource CR,  an enormous performance gain can be achieved in the identification task.

The problem of CR generation was initially introduced in \cite{part2}, where unlike in the fundamental significant papers \cite{part1}\cite{maurer}, no secrecy requirements are imposed. In particular, the CR capacity of a model involving two correlated discrete sources with one-way communication over noiseless channels with limited capacity was established. The CR capacity is defined to be the maximum rate of CR generated by two terminals using the resources available in the model\cite{part2}.

Recently, the results on CR capacity have been extended in \cite{globecom} to  SISO and point-to-point Multiple-Input Multiple-Output (MIMO) Gaussian channels, which are practically relevant in many communication situations including satellite and deep space
communication links \cite{mobile}, wired and wireless communications, etc. The results on CR capacity over Gaussian channels have been used to establish a lower-bound on its corresponding correlation-assisted  secure identification capacity in the log-log scale \cite{globecom}. This lower bound can already exceed the secure identification capacity over Gaussian channels with randomized encoding elaborated in \cite{wafapaper}. 

However, to the best of our knowledge, there are no results on the CR generation problem over fading channels. The generated CR can be exploited in the problem of correlation-assisted identification over fading channels, which is, as far as we know, an open problem.
The phenomenon of fading is one of the fundamental aspects in wireless communication. 
Fading refers to the deviation of a signal attenuation during wireless propagation with different variables such as time, rainfall, radio frequency.
A common model for wireless communication is the fading channel model with additive white Gaussian noise (AWGN) \cite{goldsmith,Tse,fadingmodels,mob2,overview}.
In our work, the focus will be on SISO slow fading channels with AWGN and with arbitrary state distribution. In the slow fading scenario, the channel state is random but remains constant over the time-scale of transmission. The event of major interest here is outage. This arises when the channel state is so poor that no coding scheme is able to establish reliable communication at a certain target rate.
 We consider, as a capacity measure of the slow fading channel, the $\eta$-outage capacity defined to be the largest rate at which one can reliably communicate over the channel with probability greater or equal to $1-\eta$ \cite{goldsmith}\cite{Tse}. To the best of our knowledge, no rigorous proof of the outage capacity for arbitrary state distribution is provided in the literature.
 
\color{black}The main contribution of this paper consists in establishing first the $\eta$-outage capacity of SISO slow fading channels with AWGN and with arbitrary state distribution. Then, we extend the concept of outage to the CR generation problem over the slow fading channel by deriving a single-letter characterization of its corresponding $\eta$-outage CR capacity using our characterization of its $\eta$-outage capacity. In the CR generation framework, outage occurs when the channel state is so poor that the terminals cannot agree on a common random variable with high probability. The $\eta$-outage CR capacity is defined to be the maximum rate of CR generated by the terminals using the resources available in the model such that the outage probability does not exceed $\eta.$\color{black}

\textit{Paper outline:} The paper is organized as follows. In Section \ref{sec2}, we present our system model, provide the key definitions and present the main results. The $\eta$-outage capacity is  established in Section \ref{sec3}.  A  rigorous proof of the $\eta$-outage 
CR capacity is provided in Section \ref{sec4}.
The conclusion contains concluding remarks.

\textit{Notation:} $\mathbb{C}$ denotes the set of complex numbers and $\mbb R$ denotes the set of real numbers; $H(\cdot)$ and $h(\cdot)$  correspond to the  entropy of discrete and continuous random variables, respectively; $I(\cdot;\cdot)$ denotes the mutual information between two random variables; $|\mathcal{K}|$ stands for the cardinality of the set $\mathcal{K}$ and $\mathcal{T}_{U}^{n}$ denotes the set of typical sequences of length $n$ and of type $P_{U}$. For any random variables $X$, $Y$ and $Z$, we use the operator $\color{black}X \circlearrow{Y} \circlearrow{Z}\color{black}$ to indicate a Markov chain. Throughout the paper, $\log$ and $\exp$ are to the base 2.
\section{System Model, Definitions and Main Results}
\label{sec2}
\subsection{System Model}
\label{systemmodel}
Let a discrete memoryless multiple source $P_{XY}$ with two components, with  generic variables $X$ and $Y$ on alphabets $\mathcal{X}$ and $\mathcal{Y}$, correspondingly, be given.
The outputs of $X$ are observed by Terminal $A$ and those of $Y$ by Terminal $B$. Both outputs have length $n.$ Terminal $A$
can send information to Terminal $B$ over the following slow fading channel $W_{G}$:
\begin{align}
z_{i}=Gt_{i}+\xi_{i} \quad i=1\hdots n.  \nonumber
\label{SISOchannelmodel}
\end{align}
where $t^n=(t_1,\hdots,t_n)\in\mbb C^n$ and $z^n=(z_1,\hdots,z_n)\in \mbb C^n$ are channel input and output blocks, respectively.
$G$ models the complex gain, where we assume that both terminals $A$ and $B$ know the distribution of the gain $G$ only. $\xi^n=(\xi_1,\hdots,\xi_n)\in \mathbb{C}^n$ models the noise sequence.
We assume that the $\xi_{i}s$ are i.i.d, where $\xi_{i} \sim \mathcal{N}_{\mathbb{C}}\left(0,\sigma^2\right),\ i=1\hdots n.$ We further assume that $G$ and $\xi^n$ are mutually independent and that $(G,\xi^n)$ are independent of $X^n$,$Y^n$. There are no other resources available to both terminals.

A CR-generation protocol of block-length $n$ consists of:
\begin{enumerate}
    \item A function $\Phi$ that maps $X^n$ into a random variable $K$ with alphabet $\mathcal{K}$ generated by Terminal $A.$
    \item A function $f$ that maps $X^n$ into some message $\ell=f(X^n).$

    \item A channel code $\Gamma$ of length $n$ for the channel $W_G$ as defined in Definition \ref{defcode}, where each codeword $\bs{t}_\ell=(t_{\ell,1},\hdots,t_{\ell,n})$  satisfies the following power constraint:
    \begin{equation}
\frac{1}{n}\sum_{i=1}^{n}t_{\ell,i}^2\leq P.   \ \
\label{energyconstraintSISOCorrelated}
\end{equation}
The random channel input sequence depending on $X^n$ is denoted by $T^n.$
    \item A function $\Lambda$ that maps $Y^n$ and the decoded message into a random variable $L$ with alphabet $\mathcal{K}$ generated by Terminal $B.$
\end{enumerate}
Such protocol induces a pair of random variable $(K,L)$ that is called permissible, where it holds for some function $\Psi$, for $D$ being the channel decoder of $\Gamma$ and for $Z^n$ being the random channel output sequence that
\begin{equation}
    K=\Phi(X^{n}), \ \     L=\Psi(Y^{n},Z^{n})=\Lambda(Y^n,D(Z^n)).
    \label{KLSISOcorrelated}
\end{equation}
This is illustrated in Fig. \ref{correlatedSISO}.
\begin{figure}[!htb]
\centering
\tikzstyle{block} = [draw, rectangle, rounded corners,
minimum height=2em, minimum width=2cm]
\tikzstyle{blockchannel} = [draw, top color=white, bottom color=white!80!gray, rectangle, rounded corners,
minimum height=1cm, minimum width=.3cm]
\tikzstyle{input} = [coordinate]
\usetikzlibrary{arrows}
\begin{tikzpicture}[scale= 1,font=\footnotesize]
\node[blockchannel] (source) {$P_{XY}$};
\node[blockchannel, below=2.5cm of source](channel) { Slow Fading channel};
\node[block, below left=1cm of source] (x) {Terminal $A$};
\node[block, below right=1cm of source] (y) {Terminal $B$};
\node[above=1cm of x] (k) {$K=\Phi(X^n)$};
\node[above=1cm of y] (l) {$L=\Psi(Y^n,Z^n)$};

\draw[->,thick] (source) -- node[above] {$X^n$} (x);
\draw[->, thick] (source) -- node[above] {$Y^n$} (y);
\draw [->, thick] (x) |- node[below right] {$T^n$} (channel);
\draw[<-, thick] (y) |- node[below left] {$Z^n$} (channel);
\draw[->] (x) -- (k);
\draw[->] (y) -- (l);

\end{tikzpicture}
\caption{Two-correlated source model with one-way communication over a  SISO slow fading channel}
\label{correlatedSISO}
\end{figure}
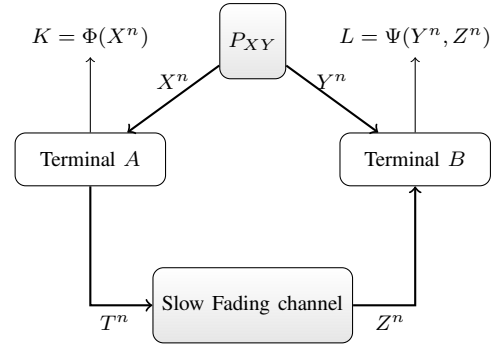
\subsection{Rates and Capacities}
We define first an achievable $\eta$-outage CR rate and the $\eta$-outage CR capacity. This is an extension to the definition of an achievable CR rate and of the CR capacity introduced in \cite{part2}.
\begin{definition} 
Fix a non-negative constant $\eta<1.$ A number $H$ is called an achievable $\eta$-outage common randomness rate  if there exists a non-negative constant $c$ such that for every $\alpha>0$ and $\delta>0$ and for sufficiently large $n$ there exists a permissible   pair of random variables $(K,L)$ such that
\begin{equation}
    \mbb P\left[\mbb P\left[K\neq L|G\right]\leq \alpha \right]\geq 1-\eta, 
    \label{errorSISOcorrelated}
\end{equation}
\begin{equation}
    |\mathcal{K}|\leq \exp(cn),
    \label{cardinalitySISOcorrelated}
\end{equation}
\begin{equation}
    \frac{1}{n}H(K)> H-\delta.
     \label{rateSISOcorrelated}
\end{equation}
\end{definition}
\begin{remark}
Together with \eqref{errorSISOcorrelated}, the technical condition  \eqref{cardinalitySISOcorrelated} ensures that for every $\epsilon>0$ and sufficiently large blocklength $n$ the set
$$\seta = \bigg\{ g \in \mbb C: \bigg| \frac{H(K|G=g)}{n}-\frac{H(L|G=g)}{n} \bigg| < \epsilon \bigg\}$$
satisfies $\mbb P\left[\seta\right]\geq 1-\eta.$ This follows from the analogous statement in\cite{part2}.
\end{remark}
\begin{definition} 
The $\eta$-outage common randomness capacity $C_{\eta,CR}(P)$ is the maximum achievable $\eta$-outage common randomness rate.
\end{definition}
Next, we define an achievable $\eta$-outage rate for the slow fading channel $W_G$ and the corresponding $\eta$-outage capacity.
For this purpose, we begin by providing the definition of a  transmission-code for  $W_G.$
\begin{definition}
\label{defcode}
A transmission-code $\Gamma$ of length $n$ and size $\lvert \Gamma \rvert$ for the channel $W_G$ is a family of pairs $\left\{(\bs{t}_\ell,\setd_\ell),\quad \ell=1,\ldots,\lvert \Gamma \rvert \right\}$ such that for all $\ell,j \in \{1,\ldots,\lvert \Gamma \rvert\}$, we have:
\begin{align}
& \bs{t}_\ell \in \mbb C^n,\quad \setd_\ell \subset \mbb C^n, \nonumber \\
&\frac{1}{n}\sum_{i=1}^{n}t_{\ell,i}^2\leq P \  \forall \ \bs{t}_\ell \in\mbb C^n,\ \bs{t}_\ell=(t_{\ell,1},\hdots,t_{\ell,n}), \nonumber \\
&\setd_\ell \cap \setd_j= \emptyset,\quad \ell \neq j. \nonumber
\end{align}
Here, $\bs{t}_\ell, \ \ell=1,\ldots,\lvert \Gamma \rvert$  and $\setd_\ell, \  \ell=1,\ldots,\lvert \Gamma \rvert,$ are the codewords and the decoding regions, respectively.
The maximum error probability is a random variable depending on $G$ and it is expressed as follows:
\begin{align}
    e(\Gamma,G)=\underset{\ell \in \{1,\ldots,\lvert \Gamma \rvert\}}{\max}W_{G}(\setd_\ell^c|\bs{t}_\ell). \nonumber
\end{align}
\end{definition}
\begin{remark}
Throughout the paper, we consider the maximum error probability criterion.
\end{remark}
\begin{definition}
    Let $0 \leq \eta<1$. A real number $R$ is called an \textit{achievable} $\eta$-\textit{outage rate} of the channel $W_G$ if for every $\theta,\delta>0$ there exists a code sequence $(\Gamma_n)_{n=1}^\infty$  such that
    \[
        \frac{\log\lvert \Gamma_n\rvert}{n}\geq R-\delta
    \]
    and
    \[
        \mbb P[e(\Gamma_n,G)\leq\theta]\geq 1-\eta
    \]
    for sufficiently large $n$.
\end{definition}
\begin{definition}
The supremum of all achievable $\eta$-outage rates is called the $\eta$-\textit{outage capacity} of the channel $W_{G}$ and is denoted by $C_\eta$.
\end{definition}
\subsection{Main Results}
In this section, we propose a single-letter characterization of the $\eta$-outage channel capacity in Theorem \ref{cetathm} and of the $\eta$-outage CR capacity in Theorem \ref{ccretathm}. Theorem \ref{cetathm} and Theorem \ref{ccretathm} are proved in Section \ref{sec3} and Section \ref{sec4}, respectively.
\begin{theorem}
    Let  $\gamma_0=\sup\{\gamma:\mbb P[\lvert G\rvert<\gamma]\leq\eta\}.$
    The $\eta$-outage capacity of the channel $W_{G}$ is equal to
 
    \[
        C_\eta(P)=\log\left(1+\frac{P\gamma_0^2}{\sigma^2}\right).
    \]
    \label{cetathm}
\end{theorem}
\begin{theorem}
For the model described in Section \ref{systemmodel}, the $\eta$-outage CR capacity is equal to
\begin{align}
C_{\eta,CR}(P) = 
  \underset{ \substack{U \\{\substack{U \circlearrow{X} \circlearrow{Y}\\ I(U;X)-I(U;Y) \leq C_{\eta}(P)}}}}{\max} I(U;X).  \nonumber
\end{align}
\label{ccretathm}
\end{theorem}
\section{Proof of Theorem \ref{cetathm}}
\label{sec3}
 Let
    \[
        \gamma_0=\sup\{\gamma:\mbb P[\lvert G\rvert<\gamma]\leq\eta\}.
    \]
 We will prove first the following lemma:
    \begin{lemma}\label{lem:leftcontinuity}
    \[
        \mbb P[|G|<\gamma_0]\leq\eta,
    \]
    so the supremum actually is a maximum.
\end{lemma}

\begin{proof}
    Let $\gamma_n\nearrow \gamma_0$ be a sequence converging to $\gamma_0$ from the left. Then 
    \[
        \{\gamma\in \mbb R:\gamma <\gamma_0\}=\bigcup_{n=1}^\infty\{\gamma\in \mbb R:\gamma <\gamma_n\}.
    \]
    From the sigma-continuity of probability measures, it follows that 
        $$\mbb P[\lvert G\rvert<\gamma_0]=\lim_n\mbb P[\lvert G\rvert<\gamma_n]\leq\eta.$$
\end{proof}
Now, we begin with the direct proof of Theorem \ref{cetathm}. We will show that
    \begin{equation}\label{eq:outage_geq}
        C_\eta(P)\geq\log\left(1+\frac{P\gamma_0^2}{\sigma^2}\right).
    \end{equation}
    Let $\theta,\delta>0$. It is well-known that there exists a code sequence $(\Gamma_n)_{n=1}^\infty$ and a blocklength $n_0$ such that
    \[
        \frac{\log\lvert \Gamma_n\rvert}{n}\geq\log\left(1+\frac{P\gamma_0^2}{\sigma^2}\right)-\delta
    \]
    and
    \[
        e(\Gamma_n,\gamma_0)\leq\theta
    \]
    for $n\geq n_0$. The degradedness of the Gaussian channels implies that also
    \[
        e(\Gamma_n,\gamma)\leq\theta
    \]
    for $n\geq n_0$, provided that $\gamma\geq \gamma_0$.  Therefore for $n\geq n_0$, Lemma \ref{lem:leftcontinuity} implies
    \[
        \mbb P[e(\Gamma_n,G)\leq\theta]
        \geq\mbb P[\lvert G\rvert\geq\gamma_0]
        =1-\mbb P[\lvert G\rvert<\gamma_0]
        \geq 1-\eta.
    \]
    This implies \eqref{eq:outage_geq}.
    
     Next, we prove the converse of Theorem \ref{cetathm}. We will show that
    \begin{equation}\label{eq:outage_leq}
        C_\eta(P)\leq\log\left(1+\frac{P\gamma_0^2}{\sigma^2}\right).
    \end{equation}
    Suppose this were not true. Then there exists an $\varepsilon>0$ such that for all $\theta,\delta>0$ there exists a code sequence $(\Gamma_n)_{n=1}^\infty$ satisfying 
    \begin{equation}\label{eq:larger}
        \frac{\log\lvert \Gamma_n\rvert}{n}\geq\log\left(1+\frac{P(\gamma_0+\varepsilon)^2}{\sigma^2}\right)-\delta
    \end{equation}
    and
    \begin{equation}\label{eq:converse_err}
        \mbb P[e(\Gamma_n,G)\leq\theta]\geq 1-\eta
    \end{equation}
    for sufficiently large $n$. The degradedness implies $e(\Gamma_n,\gamma)\geq e(\Gamma_n,\gamma_0+\varepsilon)$ for all $\gamma\leq\gamma_0+\varepsilon$. Since $\delta$ may be arbitrary, we may choose it in such a way that the right-hand side of \eqref{eq:larger} is strictly larger than $\log(1+(P\gamma_0^2)/\sigma^2)$. We define $\gamma_1$ to be the solution of the equation $$\log(1+(P\gamma_1^2)/\sigma^2)=\log\left(1+\frac{P(\gamma_0+\varepsilon)^2}{\sigma^2}\right)-\delta.$$  $\gamma_1$ is chosen such that the rate of the code sequence is greater than the capacity of the channel with gain $G$ when $|G|<\gamma_1$. Therefore, it holds for large $n$ that the error probability is greater than $\theta$ when  $|G|<\gamma_1.$ It holds for large $n$ that the error probability is greater than $\theta$ when  $|G|<\gamma_1$.
    It follows that $$
        \mbb P[e(\Gamma_n,G)>\theta]
        \geq\mbb P[\lvert G\rvert < \gamma_1]
        >\eta$$by the definition of $\gamma_0$, where we used that $\gamma_1>\gamma_0$ from the choice of $\delta$. This is a contradiction to \eqref{eq:converse_err}, and so \eqref{eq:outage_leq} must be true. This completes the proof of Theorem \ref{cetathm}.
\section{Proof of Theorem \ref{ccretathm}}
\label{sec4}
\subsection{Converse Proof}
Let $(K,L)$ be a permissible pair according to the CR-generation protocol introduced in Section \ref{systemmodel} with power constraint as in \eqref{energyconstraintSISOCorrelated}. Let $T^n$ and $Z^n$ be the random channel input and output sequence, respectively.
We further assume that $(K,L)$ satisfies $\eqref{errorSISOcorrelated}$  $\eqref{cardinalitySISOcorrelated}$ and $\eqref{rateSISOcorrelated}$.
We are going to show for $\alpha'(n)>0$ that
\begin{align}
    \frac{H(K)}{n} \leq \underset{ \substack{U \\{\substack{U \circlearrow{X} \circlearrow{Y}\\ I(U;X)-I(U;Y) \leq C_{\eta}(P)+\alpha'(n)}}}}{\max} I(U;X), \nonumber 
\end{align}
where $\underset{n\rightarrow \infty}{\lim}\alpha'(n)$ can be made arbitrarily small.
In our proof, we will use  the following lemma: 
\begin{lemma} (Lemma 17.12 in \cite{codingtheorems})
For arbitrary random variables $S$ and $R$ and sequences of random variables $X^{n}$ and $Y^{n}$, it holds that
\begin{align}
 &I(S,X^{n}|R)-I(S;Y^{n}|R)  \nonumber \\
 &=\sum_{i=1}^{n} I(S;X_{i}|X_{1}\dots X_{i-1} Y_{i+1}\dots Y_{n}R) \nonumber \\ &\quad -\sum_{i=1}^{n} I(S;Y_{i}|X_{1}\dots X_{i-1} Y_{i+1}\dots Y_{n}R) \nonumber \\
 &=n[I(S;X_{J}|V)-I(S;Y_{J}|V)],\nonumber
\end{align}
where $V=X_{1}\dots X_{J-1}Y_{J+1}\dots Y_{n}RJ$, with $J$ being a random variable independent of $R$,\ $S$, \ $X^{n}$ \ and $Y^{n}$ and uniformly distributed on $\{1 \dots n \}$.
\label{lemma1}
\end{lemma}Let $J$ be a random variable uniformly distributed on $\{1\dots n\}$ and independent of $K$, $X^n$ and $Y^n$. We further define $U=KX_{1}\dots X_{J-1}Y_{J+1}\dots Y_{n}J.$ \\
Notice  that
{{\begin{align}
H(K)&=I(K;X^{n}) \nonumber\\
&\overset{(\RM{1})}{=}\sum_{i=1}^{n} I(K;X_{i}|X_{1}\dots X_{i-1}) \nonumber\\
&=n I(K;X_{J}|X_{1}\dots X_{J-1},J) \nonumber\\
&\overset{(\RM{2})}{\leq }n I(U;X_{J}), \nonumber
\end{align}}}where $(\RM{1})$ and $(\RM{2})$ follow from the chain rule for mutual information.\\
We will show next for $\alpha'(n)>0$ that
\begin{align}
    I(U;X_J)-I(U;Y_J) \leq C_{\eta}(P)+\alpha'(n). \nonumber
\end{align}
Applying Lemma \ref{lemma1} for $S=K$, $R=\varnothing$ with $V=X_1\hdots X_{J-1}Y_{J+1}\hdots Y_{n}RJ$ yields
{{\begin{align}
&I(K;X^{n})-I(K;Y^{n}) \nonumber \\
&=n[I(K;X_{J}|V)-I(K;Y_{J}|V)] \nonumber\\
&\overset{(a)}{=}n[I(KV;X_{J})-I(K;V)-I(KV;Y_{J})+I(K;V)] \nonumber\\
&\overset{(b)}{=}n[I(U;X_{J})-I(U;Y_{J})], 
\label{UhilfsvariableSISO1}
\end{align}}}where $(a)$ follows from the chain rule for mutual information and $(b)$ follows from $U=KV$. \\
It results using (\ref{UhilfsvariableSISO1}) that
{{\begin{align}
H(K|Y^n)&=H(K)-I(K;Y^{n})\nonumber \\ &\overset{(c)}{=}H(K)-H(K|X^{n})-I(K;Y^{n}) \nonumber\\
&=I(K;X^{n})-I(K;Y^{n}) \nonumber\\
&=n[I(U;X_{J})-I(U;Y_{J})], 
\label{star2SISO2}
\end{align}}}
where $(c)$ follows because $K=\Phi(X^{n})$ from $(\ref{KLSISOcorrelated}).$ \\~\\
Let $\gamma_0=\text{sup}\Big\{\gamma:\mbb P\left[|G|<\gamma\right]\leq \eta \Big \}. $
We consider for $\epsilon>0$ being arbitrarily small the set:
\begin{align}
    &\Omega_1 \nonumber \\ &=\Big\{g \in \mathbb{C:} \   \mbb P\left[K\neq L|G=g\right]\leq \alpha \ \text{and} \ |g|\leq \gamma_0+\epsilon\Big\}.\nonumber
\end{align}
\begin{lemma}
\begin{align}
\mbb P\left[\Omega_1\right]>0. \nonumber
\end{align}
\end{lemma}
\begin{proof}
From the definition of $\gamma_0,$ we know that
\begin{align}
    \mbb P \left[|G|< \gamma_0+\epsilon\right] >\eta. \nonumber
\end{align}
This implies that
\begin{align}
    \mbb P \left[|G|\leq \gamma_0+\epsilon\right] &\geq \mbb P \left[|G|< \gamma_0+\epsilon\right] \nonumber \\
    &>\eta.
\end{align}
As a result
\begin{align}
    \mbb P\left[|G|\leq \gamma_0+\epsilon \right] =\eta_1, \nonumber
\end{align}
where $0\leq \eta<\eta_1\leq 1 .$
It follows from \eqref{errorSISOcorrelated} that
\begin{align}
    &1-\eta \nonumber \\ 
    & \leq \mbb P\left[\mbb P\left[K\neq L|G\right]\leq \alpha\right] \nonumber \\
    &=\mbb P\left[|G|\leq \gamma_0+\epsilon \right] \mbb P\left[ \mbb P\left[K\neq L\bigm|G\right]\leq \alpha\bigm||G|\leq \gamma_0+\epsilon \right] \nonumber \\
    &\quad+\mbb P\left[|G|> \gamma_0+\epsilon \right] \mbb P\left[ \mbb P\left[K\neq L\bigm|G\right]\leq \alpha\bigm||G|> \gamma_0+\epsilon \right] \nonumber \\
    &=\eta_1 \ \mbb P\left[ \mbb P\left[K\neq L\bigm|G\right]\leq \alpha\bigm||G|\leq \gamma_0+\epsilon \right] \nonumber \\
    &\quad+(1-\eta_1) \ \mbb P\left[ \mbb P\left[K\neq L\bigm|G\right]\leq \alpha\bigm||G|> \gamma_0+\epsilon \right]
    \nonumber \\
    &\leq \mbb P\left[ \mbb P\left[K\neq L\bigm|G\right]\leq \alpha\bigm||G|\leq \gamma_0+\epsilon \right]+(1-\eta_1) \nonumber \\
    &<\mbb P\left[ \mbb P\left[K\neq L\bigm|G\right]\leq \alpha\bigm||G|\leq \gamma_0+\epsilon \right]+(1-\eta), \nonumber
\end{align}
where we used that $1-\eta_1<1-\eta.$ This means that
\begin{align}
    \mbb P\left[ \mbb P\left[K\neq L\bigm|G\right]\leq \alpha\bigm||G|\leq \gamma_0+\epsilon \right]>0. \nonumber
\end{align}
In addition, since $\eta_1>0$, it follows that
\begin{align}
    \mbb P\left[ \mbb P\left[K\neq L\bigm|G\right]\leq \alpha,|G|\leq \gamma_0+\epsilon \right]>0. \nonumber
\end{align}
It follows that
$$\mbb P\left[\Omega_1\right]>0.$$
\end{proof}Next, we define $\tilde{G}$ to be a  random variable, independent of $X^n$,$Y^n$ and $\xi^n$, with alphabet $\Omega_1$ such that for every Borel set $\seta \subseteq \mbb C,$ it holds that
\begin{align}
    \mbb P \left[ \tilde{G}\in \seta \right]=\mbb P \left[ G\in\seta|G\in\Omega_1\right]. \nonumber
\end{align}
We fix the CR generation protocol and change the state distribution of the slow fading channel. We obtain the following new channel:
\begin{align}
    \tilde{Z}_{i}=\tilde{G}T_i+\xi_i \quad i=1\hdots n,  \nonumber
\end{align}
where $\tilde{Z}^n$ is the new output sequence.
We further define $\tilde{L}$ such that
\begin{align}
    \tilde{L}=\Psi(Y^n,\tilde{Z}^n). \nonumber
\end{align}
Clearly, it holds that
\begin{align}
     \mbb P\left[K\neq \tilde{L}|\tilde{G}=g\right]\leq \alpha  \quad \forall g \in \Omega_1,
     \label{newerrorinequality}
\end{align}
and that
\begin{align}
    \log(1+\frac{|g|^2P}{\sigma^2})\leq \log(1+\frac{(\gamma_0+\epsilon)^2P}{\sigma^2})\quad \forall g \in \Omega_1.
    \label{absolutevalue}
\end{align}
Furthermore, since $\xi_i \sim \mathcal{N}_{\mbb C}(0,\sigma^2), i=1\hdots n$,  it follows from \eqref{energyconstraintSISOCorrelated} that for $i=1\hdots n,$
\begin{align}
    I(T_{i},\tilde{Z}_{i}|\tilde{G}=g)\leq \log(1+\frac{|g|^2P}{\sigma^2}) \quad \forall \ g \in \Omega_1.
    \label{mutualinfmax}
\end{align}
We have:
\begin{align}
    H(K|Y^n)&=H(K|\tilde{G},Y^n) \nonumber \\
            &=H(K|\tilde{G},Y^n,\tilde{Z}^n)+I(K;\tilde{Z}^n|\tilde{G},Y^n), \nonumber
\end{align}
where we used that $\tilde{G}$ is independent of $(K,Y^n).$ 
On the one hand, we have:
\begin{align}
                                H\left(K|\tilde{Z}^n,\tilde{G},Y^n\right)  &\overset{(a)}{\leq } H\left(K|\tilde{L},\tilde{G}\right) \nonumber \\    &\overset{(b)}{\leq } \mbb E\left[1+\log|\mathcal{K}|\mbb P[K\neq \tilde{L}|\tilde{G}]\right]  \nonumber \\                          &=1+ \log|\mathcal{K}|\mbb E \left[ P[K\neq \tilde{L}|\tilde{G}]\right]  \nonumber \\
                                &\overset{(c)}{\leq } 1+\alpha \log|\mathcal{K}|  \nonumber \\                               &\overset{(d)}{\leq } 1+\alpha \ c n, \nonumber
\end{align}
where (a) follows from $\tilde{L}=\Psi(Y^n,\tilde{Z}^n)$, (b) follows from Fano's Inequality, (c) follows from \eqref{newerrorinequality}  and (d) follows from $\log|\mathcal{K}|\leq cn$ in \eqref{cardinalitySISOcorrelated}.

    On the other hand, we have:
   {{\begin{align} 
I(K;\tilde{Z}^n|\tilde{G},Y^{n})&\leq I(X^{n}K;\tilde{Z}^n|\tilde{G},Y^{n}) \nonumber\\
& \overset{(a)}{\leq } I(T^{n};\tilde{Z}^n|\tilde{G},Y^{n})  \nonumber \\
& = h(\tilde{Z}^n|\tilde{G},Y^{n})- h(\tilde{Z}^n|T^{n},\tilde{G},Y^{n}) \nonumber \\
& \overset{(b)}{=} h(\tilde{Z}^n|\tilde{G},Y^{n})- h(\tilde{Z}^n|\tilde{G},T^{n}) \nonumber \\
& \overset{(c)}{\leq }  h(\tilde{Z}^n|\tilde{G})- h(\tilde{Z}^n|\tilde{G},T^{n}) \nonumber \\
& = I(T^{n};\tilde{Z}^n|\tilde{G})  \nonumber \\
& \overset{(d)}{=} \sum_{i=1}^{n} I(\tilde{Z}_{i};T^{n}|\tilde{G},\tilde{Z}^{i-1}) \nonumber \\
& = \sum_{i=1}^{n} h(\tilde{Z}_{i}|\tilde{G},\tilde{Z}^{i-1})-h(\tilde{Z}_{i}|\tilde{G},T^{n},\tilde{Z}^{i-1}) \nonumber \\
& \overset{(e)}{=} \sum_{i=1}^{n} h(\tilde{Z}_{i}|\tilde{G},\tilde{Z}^{i-1})-h(\tilde{Z}_{i}|\tilde{G},T_{i}) \nonumber \\
& \overset{(f)}{\leq} \sum_{i=1}^{n} h(\tilde{Z}_{i}|\tilde{G})-h(\tilde{Z}_{i}|\tilde{G},T_{i}) \nonumber \\
& = \sum_{i=1}^{n} I(T_{i};\tilde{Z}_{i}|\tilde{G}) \nonumber \\
&\overset{(g)}{\leq}n\mbb E \left[\log(1+\frac{|\tilde{G}|^2P}{\sigma^2})\right]\nonumber \\
& \overset{(h)}{\leq} n \log(1+\frac{(\gamma_0+\epsilon)^2P}{\sigma^2}) \nonumber \\
&\overset{(i)}{=}n (C_{\eta}(P)+\epsilon'),\nonumber
\end{align}}}with $\epsilon'$ being arbitrarily small, where $(a)$ follows from the Data Processing Inequality because $Y^{n}\circlearrow{X^{n}K}\circlearrow{\tilde{G}T^{n}}\circlearrow{\tilde{Z}^{n}}$ forms a Markov chain, $(b)$ follows because $Y^{n}\circlearrow{X^{n}K}\circlearrow{\tilde{G}T^{n}}\circlearrow{\tilde{Z}^{n}}$ forms a Markov chain, $(c)(f)$ follow because conditioning does not increase entropy, $(d)$ follows from the chain rule for mutual information, $(e)$ follows because $T_{1}\dots T_{i-1} T_{i+1}\dots T_{n}\tilde{Z}^{i-1} \circlearrow{\tilde{G}T_{i}}\circlearrow{\tilde{Z}_{i}}$ forms a Markov chain, $(g)$ follows from \eqref{mutualinfmax} and
$(h)$ follows from \eqref{absolutevalue} and $(i)$ follows from Theorem \ref{cetathm} using that $\epsilon$ is arbitrarily small.\\
This proves that for $0\leq \eta<1$
 \begin{align}
     \frac{H(K|Y^n)} {n} \leq C_{\eta}(P)+\alpha'(n),
     \label{star1SISOnonemptySISO}
 \end{align}
 where $\alpha'(n)=\frac{1}{n}+\alpha c+\epsilon' >0.$ 
 \\~\\
From (\ref{star2SISO2}) and (\ref{star1SISOnonemptySISO}), we deduce that for $0\leq\eta<1$
{{\begin{align}
&I(U;X_{J})-I(U;Y_{J})  \leq C_{\eta}(P) +\alpha'(n), \nonumber
\end{align}}}where $U \circlearrow{X_{J}} \circlearrow{Y_{J}}.$ \\
\color{black}Since the joint distribution of $X_{J}$ and $Y_{J}$ is equal to $P_{XY}$, $\frac{H(K)}{n}$ is upper-bounded by $I(U;X)$ subject to $I(U;X)-I(U;Y) \leq C_{\eta}(P) + \alpha'(n)$ with $U$ satisfying $U \circlearrow{X} \circlearrow{Y}$. As a result, for $\alpha'(n)>0$, it holds that
\begin{align}
    \frac{H(K)}{n} \leq \underset{ \substack{U \\{\substack{U \circlearrow{X} \circlearrow{Y}\\ I(U;X)-I(U;Y) \leq C_{\eta}(P)+\alpha'(n)}}}}{\max} I(U;X).
    \nonumber
\end{align}
Here, $\underset{n\rightarrow\infty}{\lim}\alpha'(n)$ can be made arbitrarily small. This completes the converse proof.
\subsection{Direct Proof}
\label{SISOproof}
We extend the coding scheme provided in \cite{part2} to slow fading channels. By continuity, it suffices to show that 
$$ C'_{\eta,CR}(P)=\underset{ \substack{U \\{\substack{U \circlearrow{X} \circlearrow{Y}\\ I(U;X)-I(U;Y) \leq C'}}}}{\max} I(U;X)  $$ is an achievable $\eta$-outage CR rate for every $C'<C_{\eta}(P).$
Let $U$ be a random variable satisfying $I(U;X)-I(U;Y) \leq C'$. We are going to show that $H=I(U;X)$ is an achievable $\eta$-outage CR rate. Without loss of generality, assume that the distribution of $U$ is a possible type for block length $n$.
Let
{{\begin{align}
N_{1}&=\exp\left(n[I(U;X)-I(U;Y)+3\delta]\right) \nonumber\\
N_{2}&=\exp\left(n[I(U;Y)-2\delta]\right). \nonumber
\end{align}}}For each pair $(i,j)$ with $1\leq i \leq N_{1}$ and $1\leq j \leq N_{2}$, we define a random sequence $\bs{U}_{i,j}\in\mathcal{U}^n$ of type $P_{U}$. Each realization $\bs{u}_{i,j}$ of $\bs{U}_{i,j}$ is known to both terminals.  
This means that  $N_{1}$ codebooks $C_{i}, 1\leq i \leq N_{1}$, are known to both terminals, where each codebook contains $N_{2}$ sequences $ \bs{u}_{i,j}, \  j=1\hdots N_2$.\\
It holds for every $X$-typical $\bs{x}$ that $$\mbb P[\exists (i,j) \ \text{s.t} \ \bs{U}_{ij} \in \mathcal{T}_{U|X}^{n}\left(\bs{x}\right)|X^{n}=\bs{x}]\geq 1-\exp(-\exp(nc')),$$ for a suitable $c'>0$, as in the proof of Theorem 4.1 of \cite{part2}. 
 For $K(\bs{x})$, we choose a sequence $\bs{u}_{ij}$ jointly typical with $\bs{x}$ (either one if there are several).   Let $f(\bs{x})=i$ if $K(\bs{x}) \in C_{i}$.  If no such $\bs{u}_{i,j}$ exists, then $f(\bs{x})=N_1+1$ and $K(\bs{x})$ is set to a constant sequence $\bs{u}_0$ different from all the $\bs{u}_{ij}s$ and known to both terminals.
  Since $ C'<C_{\eta}(P)$, we choose $\delta$ to be sufficently small such that
      \begin{align}
     \frac{\log \lVert f \rVert}{n}&=\frac{\log(N_1+1)}{n} \nonumber \\
     &\leq C_{\eta}(P)-\delta',
     \label{inequalitylogfSISO}
      \end{align}
 for some $\delta'>0$, where $\lVert f \rVert$ refers to the cardinality of the set of messages $\{i^\star=f(\bs{x})\}.$ This is the same notation used in \cite{codingtheorems}.
 The message $i^\star=f(\bs{x})$, with $i^\star\in\{1,\hdots,N_1+1\}$, is encoded to a sequence $\bs{t}$ using a code sequence $(\Gamma^\star_n)_{n=1}^{\infty}$ with rate $\frac{\log \lvert \Gamma^\star_n \rvert}{n}=\frac{\log \lVert f \rVert}{n}$ satisfying \eqref{inequalitylogfSISO}
 and with error probability $e(\Gamma^\star_n,G)$ satisfying:
 \begin{align}
     \mbb P\left[e(\Gamma^\star_n,G) \leq \theta \right] \geq 1-\eta, 
     \label{outageprobCR}
 \end{align}
where $\theta$ is sufficiently small for sufficiently large $n$.
  \color{black}
  From the definition of the $\eta$-outage capacity, we know that such a code sequence exists. The sequence $\bs{t}$ is sent over the slow fading channel. Let $\bs{z}$ be the channel output sequence. Terminal $B$ decodes the message $\tilde{i}^\star$ from the knowledge of $\bs{z}.$
Let $L(\bs{y},\tilde{i}^\star)=\bs{u}_{\tilde{i}^\star,j}$ if $\bs{u}_{\tilde{i}^\star,j}$ and $\bs{y}$ are UY-typical . If there is no such $\bs{u}_{\tilde{i}^\star,j}$ or there are several, $L$ is set equal to $\bs{u}_0$ (since $K$ and $L$ must have the same alphabet).
Now, we are going to show that the requirements in $\eqref{errorSISOcorrelated}$  $\eqref{cardinalitySISOcorrelated}$ and $\eqref{rateSISOcorrelated}$ are satisfied.
Clearly, (\ref{cardinalitySISOcorrelated}) is satisfied  for $c=2(H(X)+1)$, $n$ sufficiently large:
{{\begin{align}
|\mathcal{K}|&=N_1 N_2+1 \nonumber \\
             &= \exp(n\left[I(U;X)+\delta\right])+1 \nonumber \\
             &\leq\exp(2n\left[H(X)+1\right]).\nonumber
\end{align}}}We define next for a fixed $\bs{u}_{i,j}$ the set
$$\Omega=\{ \bs{x}\in\mathcal{X}^{n} \ \text{s.t.} \ (\bs{x},\bs{u}_{i,j}) \ \text{jointly typical}\}.$$
As shown in \cite{part2}, it holds that
\begin{align}
&\mbb P[K=\bs{u}_{i,j}] \nonumber \\&=\sum_{\bs{x}\in\Omega}\mbb P[K=\bs{u}_{i,j}|X^n=\bs{x}]P_{X}^n(\bs{x}) \nonumber \\
&\quad+\sum_{\bs{x}\in\Omega^c}\mbb P[K=\bs{u}_{i,j}|X^n=\bs{x}]P_{X}^n(\bs{x}) \nonumber \\
&\overset{(\RM{1})}{=}\sum_{\bs{x}\in\Omega}\mbb P[K=\bs{u}_{i,j}|X^n=\bs{x}]P_{X}^n(\bs{x}) \nonumber \\
&\leq \sum_{\bs{x}\in\Omega}P_{X}^n(\bs{x}) \nonumber \\
&=P_{X}^{n}(\{\bs{x}: (\bs{x},\bs{u}_{i,j}) \ \text{jointly typical}\}) \nonumber\\
& = \exp\left(-nI(U;X)+o(n)\right), \nonumber
\end{align}
where (\RM{1}) follows because for  $(\mathbf{x},\bs{u}_{i,j})$ being not jointly typical, we have $\mbb P[K=\bs{u}_{i,j}|X^n=\bs{x}]=0.$\\
This yields
{{\begin{align}
H(K) & \geq 
n I(U;X)+ o(n) \nonumber\\
& = n H+o(n). \nonumber
\end{align}}}Thus, (\ref{rateSISOcorrelated}) is satisfied.
Now, it remains to prove that \eqref{errorSISOcorrelated} is satisfied. Let $\bs{M}=\bs{U}_{11}\hdots \bs{U}_{N_{1}N_{2}}.$ We define the following two sets which depend on $\bs{M}$:
\begin{align}
    S_{1}(\bs{M})&=\{(\bs{x},\bs{y}):(K(\bs{x}),\bs{x},\bs{y}) \in \mathcal{T}_{UXY}^{n}\} \nonumber
\end{align} and
\begin{align}
    S_{2}(\bs{M})=\{(\bs{x},\bs{y}):(\bs{x},\bs{y}) \in S_{1}(\bs{M}) \ \text{s.t.} \ \exists \ \bs{U}_{i\ell}\neq\bs{U}_{ij}=K(\bs{x}) \nonumber \\   \text{jointly typical with} \ \bs{y} \ (\text{with the same first index} \ i)
\}.\nonumber
\end{align}
It is proved in \cite{part2} that
\begin{align}
    \mathbb{E}_{\bs{M}}\left[P_{XY}^{n}(S_{1}^{c}(\bs{M}))+P_{XY}^{n}(S_{2}(\bs{M}))\right] \leq \beta,
    \label{averagebeta}
\end{align}
where $\beta$ is exponentially small for sufficiently large $n$. \begin{remark}
$P_{XY}^{n}(S_{1}^{c}(\bs{M}))$ and $P_{XY}^{n}(S_{2}(\bs{M}))$ are here random variables depending on $\bs{M}.$
\end{remark}We choose a realization $\bs{m}=\bs{u}_{11}\hdots \bs{u}_{N_1N_2}$ satisfying:
\begin{align}
    P_{XY}^{n}(S_{1}^{c}(\bs{m}))+P_{XY}^{n}(S_{2}(\bs{m})) \leq \beta. \nonumber
\end{align} 
From \eqref{averagebeta}, we know that such a realization exists.
Now, we define the following event:
\begin{align}
    \mathcal{D}_{\bs{m}}= ``K(X^n) \ \text{is equal to none of the} \  \bs{u}_{i,j}s". \nonumber
\end{align}
We further define $I^\star=f(X^n)$ to be the random variable modeling the message encoded by Terminal $A$ and  $\tilde{I}^\star$ to be the random variable modeling the message decoded by Terminal $B$. 
We have:
\begin{align}
    \mbb P[K\neq L|G] \nonumber &=\mbb P[K\neq L|G,I^\star=\tilde{I}^\star]\mbb P[I^\star=\tilde{I}^\star|G] \nonumber \\
        &\quad + \mbb P[K\neq L|G,I^\star\neq \tilde{I}^\star]\mbb P[I^\star\neq\tilde{I}^\star|G] \nonumber \\
        &\leq \mbb P[K\neq L|G,I^\star=\tilde{I}^\star]+ \mbb P[I^\star\neq\tilde{I}^\star|G]\nonumber.
\end{align}
Here:
\begin{align}
   & \mbb P[K\neq L|G,I^\star=\tilde{I}^\star] \nonumber \\
   &= \mbb P[K\neq L|G,I^\star=\tilde{I}^\star,\mathcal{D}_{\bs{m}}]\mbb P[\mathcal{D}_{\bs{m}}|G,I^\star=\tilde{I}^\star]   \nonumber \\
   &\quad + \mbb P[K\neq L|G,I^\star=\tilde{I}^\star,\mathcal{D}_{\bs{m}}^c]\mbb P[\mathcal{D}_{\bs{m}}^c|G,I^\star=\tilde{I}^\star] \nonumber \\
   &\overset{(\RM{1})}{=}\mbb P[K\neq L|G,I^\star=\tilde{I}^\star,\mathcal{D}_{\bs{m}}^c]\mbb P[\mathcal{D}_{\bs{m}}^c|G,I^\star=\tilde{I}^\star] \nonumber \\
   &\leq \mbb P[K\neq L|G,I^\star=\tilde{I}^\star,\mathcal{D}_{\bs{m}}^c],\nonumber
\end{align}
where $(\RM{1})$ follows from $\mbb P[K\neq L|G,I^\star=\tilde{I}^\star,\mathcal{D}_{\bs{m}}]=0,$ since conditioned on $G$, $I^\star=\tilde{I}^\star$ and $\mathcal{D}_{\bs{m}}$, we know that $K$ and $L$ are both equal to $\bs{u}_0$.
It follows that
\begin{align}
    &\mbb P[K\neq L|G] \nonumber \\
    &\leq \mbb P[K\neq L|G,I^\star=\tilde{I}^\star,\mathcal{D}_{\bs{m}}^c]+ \mbb P[I^\star\neq\tilde{I}^\star|G] \nonumber \\
    &\leq P_{XY}^n\left(S_{1}^{c}(\bs{m})\cup S_{2}(\bs{m})\right)+ \mbb P[I^\star\neq\tilde{I}^\star|G] \nonumber \\
    &\overset{(a)}{\leq}P_{XY}^{n}(S_{1}^{c}(\bs{m}))+P_{XY}^n\left(S_{2}(\bs{m})\right) +\mbb P[I^\star\neq\tilde{I}^\star|G] \nonumber \\
    &\leq \beta+ \mbb P[I^\star\neq\tilde{I}^\star|G],\nonumber
\end{align}
where $(a)$ follows from the union bound. \\
From \eqref{outageprobCR}, we know that
\begin{align}
   \mbb P\left[ \mbb P\left[I^\star\neq \tilde{I}^\star|G\right]\leq \theta\right] \geq 1-\eta. \nonumber
\end{align}
We have:
\begin{align}
    &\mbb P\left[I^\star\neq \tilde{I}^\star|G\right]\leq \theta \implies \mbb P[K\neq L|G] \leq \beta+ \theta. \nonumber
\end{align}
By choosing $\alpha=\beta+ \theta$, we have:
\begin{align}
    &\mbb P\left[I^\star\neq \tilde{I}^\star|G\right]\leq \theta  \implies \mbb P[K\neq L|G] \leq \alpha. \nonumber
\end{align}
Thus:
\begin{align}
    \mbb P\left[  \mbb P[K\neq L|G] \leq \alpha          \right] &\geq \mbb P\left[\mbb P\left[ I^\star\neq \tilde{I}^\star|G     \right]\leq \theta\right] \nonumber \\
    &\geq 1-\eta. \nonumber
\end{align} 
Here, $\alpha$ is arbitrarily small for sufficiently large $n$.
This completes the direct proof.
\section{Conclusion}
In this paper, we have examined the problem of common randomness generation over slow fading channels for their practical relevance in many situations in wireless communications. The generated CR can be exploited in the identification scheme to improve the performance gain. We established a single-letter characterization of the outage CR capacity over slow fading channels with AWGN and with arbitrary state distribution using our characterization of its corresponding channel outage capacity. 
As a future work, it would be interesting to study the problem of CR generation over single-user MIMO slow fading channels since it is known that, compared  to  SISO  systems, point-to-point MIMO communication systems offer higher rates,  more reliability and resistance to interference.
Future research might also focus on studying the problem of CR generation over fast fading channels.

\section*{Acknowledgments} 
We thank the German Research Foundation (DFG) within the Gottfried Wilhelm Leibniz Prize under Grant BO 1734/20-1 for their support of H. Boche and M. Wiese.
Thanks also go to the German Federal Ministry of Education and Research (BMBF) within the national initiative for “Post Shannon Communication (NewCom)” with the project “Basics, simulation and demonstration for new communication models” under Grant 16KIS1003K for their support of H. Boche, R. Ezzine and with the project “Coding theory and coding methods for new communication models” under Grant 16KIS1005 for their support of C. Deppe.
Further, we thank the German Research Foundation (DFG) within Germany’s Excellence Strategy EXC-2111—390814868 and EXC-2092 CASA - 390781972 for their support of H. Boche and M. Wiese.
\bibliographystyle{IEEEtran}
\bibliography{definitions,references}

\end{document}